\documentclass[11pt,a4paper]{amsart}

\usepackage{amssymb, amsmath, amsthm, enumerate, fancyhdr} 
\usepackage[utf8]{inputenc}
\usepackage[T1]{fontenc}
\usepackage{ifthen}
\usepackage{tikz}
\usetikzlibrary{shapes,fit,automata,decorations.pathmorphing}

\allowdisplaybreaks

\usepackage{cancel, cases, mathabx, fullpage}

\newcommand{\makePageBig}{
\setlength{\topmargin}{-1cm}
\setlength{\oddsidemargin}{-7mm}
\setlength{\evensidemargin}{-7mm}
\textwidth18cm
\textheight235mm
\headheight14pt
\headsep5mm}

\providecommand{\ignore}[1]{}

\newcommand{\Txt}[1][]{\mathbf{Txt}^{#1}}

\newcommand{\It}{\mathbf{It}}

\newcommand{\BMS}{\mathbf{BMS}}

\newcommand{\Ex}{\mathbf{Ex}}

\newcommand{\Bc}{\mathbf{Bc}}
\newcommand{\Cons}{\mathbf{Cons}}
\newcommand{\Comp}{\mathbf{Cons}}

\newcommand{\Conv}{\mathbf{Conv}}
\newcommand{\Caut}{\mathbf{Caut}}

\newcommand{\NU}{\mathbf{NU}}
\newcommand{\SNU}{\mathbf{SNU}}
\newcommand{\Dec}{\mathbf{Dec}}
\newcommand{\SDec}{\mathbf{SDec}}
\newcommand{\Wb}{\mathbf{Wb}}

\newcommand{\WMon}{\mathbf{WMon}}
\newcommand{\Mon}{\mathbf{Mon}}
\newcommand{\SMon}{\mathbf{SMon}}

\newcommand{\natnum}{\mathbb{N}}
\newcommand{\N}{\mathbb{N}}

\newcommand{\Seq}[1]{\text{$#1^{<\omega}$}}
\newcommand{\ISeq}[1]{\text{$#1^{\omega}$}}

\newcommand{\concat}{{^\smallfrown}}
\newcommand{\ps}{\mathrm{pos}}

\newcommand{\pump}{\mathrm{pump}}
\newcommand{\visit}{V}
\newcommand{\states}{\mathrm{pr}_1}
\newcommand{\Mypath}{\mathrm{path}}

\newcommand{\totalCp}{\mathcal{R}}

\newcommand{\pad}{\mathrm{pad}}

\newcommand{\last}{\mathrm{last}}
\newcommand{\dom}{\mathrm{dom}}
\newcommand{\ran}{\mathrm{ran}}
\newcommand{\cnt}{\mathrm{content}}

\newcommand{\simu}{\mathfrak{s}}
\newcommand{\Simu}{\mathfrak{S}}

\newcommand{\NoMC}{\mathrm{NoMC}}
\newcommand{\MC}{\mathrm{MC}}

\newcommand{\pr}{\mathrm{pr}}

\newcommand{\sqr}[2]{{\vcenter{\hrule height.#2pt
        \hbox{\vrule width.#2pt height#1pt \kern#1pt
           \vrule width.#2pt}
        \hrule height.#2pt}}}

\newcommand{\Qed}[1]{{\nobreak\hfil\penalty50
  \hskip1em\hbox{}\nobreak\hfil$\sqr{8}{8}$\enspace\sc #1
  \parfillskip=0pt \finalhyphendemerits=0 \par}
  \bigskip}

\newcommand{\QedFor}[1]{\Qed{(for~#1)}}
\newcommand{\QedForClaim}[1][]{\QedFor{~Claim#1}}

\renewenvironment{proof}{

\noindent
\noindent\emph{Proof.}\enspace}{\Qed{}}

\newcommand{\CalL}{\mathcal{L}}

\newtheorem{thm}{Theorem}[section]

\newtheorem{cor}[thm]{Corollary}
\newtheorem{lem}[thm]{Lemma}

\theoremstyle{definition}
\newtheorem{defn}[thm]{Definition}

\newcounter{claimCounter}[thm]

\newcounter{caseCounter}[thm]

\newcounter{subcaseCounter}[caseCounter]

\newboolean{useproof}
\setboolean{useproof}{false}

\def\makeinnocent#1{\catcode`#1=12 }
\def\csarg#1#2{\expandafter#1\csname#2\endcsname}

\def\ThrowAwayComment#1{\begingroup
    \def\CurrentComment{#1}    \let\do\makeinnocent \dospecials
    \makeinnocent\^^L    \endlinechar`\^^M \catcode`\^^M=12 \xComment}
{\catcode`\^^M=12 \endlinechar=-1  \gdef\xComment#1^^M{\def\test{#1}
      \csarg\ifx{PlainEnd\CurrentComment Test}\test
          \let\next\endgroup
      \else \csarg\ifx{LaLaEnd\CurrentComment Test}\test
            \edef\next{\endgroup\noexpand\end{\CurrentComment}}
      \else \let\next\xComment
      \fi \fi \next}
}

{\escapechar=-1\relax
	\csarg\xdef{PlainEndproofTest}{\string\\endproof}	\csarg\xdef{LaLaEndproofTest}{\string\\end\string\{proof\string\}}}

\renewcommand{\proof}{\ifthenelse{\boolean{useproof}}{

\noindent
\noindent\emph{Proof.}}{\ThrowAwayComment{proof}}}

\renewcommand{\endproof}{\ifthenelse{\boolean{useproof}}{\Qed{}}{}}

\makePageBig

\newcommand{\separate}{sep}

\title[Learners for Languages from Text with Finitely Many Memory Changes]{Learning Languages in the Limit from Positive Information \\ with Finitely Many Memory Changes}

\author{%
 Timo K{\"o}tzing, Karen Seidel
}

\tikzset{>=latex}

\setboolean{useproof}{true}

\begin{document}

\begin{abstract}
We investigate learning collections of languages from texts by an inductive inference machine with access to the current datum and a bounded memory in form of states.
Such a bounded memory states ($\BMS$) learner is considered successful in case it eventually settles on a correct hypothesis while exploiting only finitely many different states.

We give the complete map of all pairwise relations for an established collection of criteria of successfull learning.
Most prominently, we show that non-U-shapedness is not restrictive, while conservativeness and (strong) monotonicity are.
Some results carry over from iterative learning by a general lemma showing that, for a wealth of restrictions (the \emph{semantic} restrictions), iterative and bounded memory states learning are equivalent.
We also give an example of a non-semantic restriction (strongly non-U-shapedness) where the two settings differ.
\end{abstract}

\maketitle


\setboolean{useproof}{true}

\newcommand{\justify}[1]{\textcolor{red}{#1}}

\newcommand{\drawBackboneTwo}{
\begin{scope}[every node/.style={minimum size=4mm}]

\node (nothing) at (0,0) {\textbf{T}};
\node (nu)  at (3,-1)       {$\NU$};
\node (dec)  at (1,-3)      {$\Dec$};

\node (smon) at (-5,-8)     {$\SMon$};
\node (mon) at (-5,-6)      {$\Mon$}; 
\node (wmon) at (-1.5,-3.5)     {$\WMon$};
\node (caut) at (-2,-2)     {$\Caut$};

\node (sdec) at (3,-5)      {$\SDec$};
\node (snu) at (5,-3)       {$\SNU$};
\node (conv) at (5,-7)      {$\Conv$};

\draw[<-] (dec) -- (smon);
\draw[<-] (nu) -- (snu);
\draw[<-] (nu) -- (dec);
\draw[<-] (dec) -- (sdec); 
\draw[<-] (snu) -- (sdec);
\draw[<-] (snu) -- (conv);
\draw[<-] (nothing) -- (caut);
\draw[<-] (caut) -- (smon);
\draw[<-] (nothing) -- (nu);
\draw[<-] (nothing) -- (wmon);
\draw[<-] (nothing) edge[bend right] (mon);
\draw[<-] (mon) -- (smon);
\draw[<-] (wmon) edge[bend right=25] (conv);
\draw[<-] (wmon) -- (smon);
\end{scope}
}

\section{Introduction}

We are interested in the problem of algorithmically learning a description for a formal language (a computably enumerable subset of the set of natural numbers) when presented successively all and only the elements of that language; this is sometimes called \emph{inductive inference}, a branch of (algorithmic) learning theory. For example, a learner $M$ might be presented more and more even numbers. After each new number, $M$ outputs a description for a language as its conjecture. The learner $M$ might decide to output a program for the set of all multiples of $4$, as long as all numbers presented are divisible by~$4$. Later, when $M$ sees an even number not divisible by $4$, it might change this guess to a program for the set of all multiples of~$2$.

\medskip
Many criteria for deciding whether a learner $M$ is \emph{successful} on a language~$L$ have been proposed in the literature. Gold, in his seminal paper \cite{Gol:j:67}, gave a first, simple learning criterion, \emph{$\Txt\Ex$-learning}\footnote{$\Txt$ stands for learning from a \emph{text} of positive examples; $\Ex$ stands for \emph{explanatory}.}, where a learner is \emph{successful} iff, on every \emph{text} for $L$ (listing of all and only the elements of $L$) it eventually stops changing its conjectures, and its final conjecture is a correct description for the input sequence.
Trivially, each single, describable language $L$ has a suitable constant function as an $\Txt\Ex$-learner (this learner constantly outputs a description for $L$). Thus, we are interested in analyzing for which \emph{classes of languages} $\CalL$ there is a \emph{single learner} $M$ learning \emph{each} member of $\CalL$. Sometimes, this framework is called \emph{language learning in the limit} and has been studied extensively, using a wide range of learning criteria similar to $\Txt\Ex$-learning (see, for example, the textbook \cite{Jai-Osh-Roy-Sha:b:99:stl2}).

One major criticism of the model suggested by Gold is its excessive use of memory, see for example \cite{Cas-Moe:j:08:NUIt}: for each new hypothesis the entire history of past data is available.
Iterative learning is the most common variant of learning in the limit which addresses memory constraints: the memory of the learner on past data is just its current hypothesis.
Due to the padding lemma \cite{Jai-Osh-Roy-Sha:b:99:stl2}, this memory is not necessarily void, but only finitely many data can be memorized in the hypothesis.
There is a comprehensive body of work on iterative learning, see, e.g., \cite{Cas-Koe:c:10:colt,Cas-Moe:j:08:NUIt,%
jain2016role,Jai-Moe-Zil:j:13,Jai-Osh-Roy-Sha:b:99:stl2}. 

Another way of modelling restricted memory learning is to grant the learner access to not their current hypothesis, but a \emph{state} which can be used in the computation of the next hypothesis (and next state). This was introduced in \cite{Car-Cas-Jai-Ste:j:07} and called \emph{bounded memory states (BMS)} learning.
It is a reasonable assumption to have a countable reservoir of states. Assuming a computable enumeration of these states, we use natural numbers to refer to them.
Note that allowing arbitrary use of all natural numbers as states would effectively allow a learner to store all seen data in the state, thus giving the same mode as Gold's original setting.
Probably the minimal way to restrict the use of states is to demand for successful learning that a learner must stop using new states eventually (but may still traverse among the finitely many states produced so far, and may use infinitely many states on data for a non-target language).
It was claimed that this setting is equivalent to iterative learning \cite[Remark~38]{Car-Cas-Jai-Ste:j:07} (this restriction is called \emph{ClassBMS} there, we refer to it by $\Txt\BMS_\ast\Ex$).
However, this was only remarked for the plain setting of explanatory learning; for further restrictions, the setting is completely unknown, only for explicit constant state bounds a few scattered results are known, see \cite{Car-Cas-Jai-Ste:j:07,Cas-Koe:j:11:memLess}.

\medskip
In this paper, we consider a wealth of restrictions, described in detail in Section~\ref{sec:prelims} (after an introduction to the general notation of this paper).
Following the approach of giving \emph{maps} of pairwise relations suggested in \cite{kotzing2016towards}, we give a complete map in Figure~\ref{fig:TxtBMSEx}.
We note that this map is the same as the map for iterative learning given in \cite{jain2016role}, but partially for different reasons.

In Lemma~\ref{BMS-It_sem-afsoet} we show that, for many restrictions (the so-called \emph{semantic} restrictions, where only the semantics of hypotheses are restricted) the learning setting with bounded memory states is equivalent to learning iteratively.
This proves and generalizes the aforementioned remark in \cite{Car-Cas-Jai-Ste:j:07} to a wide class of restrictions.
The iterative learner uses the hypotheses of the $\BMS_\ast$-learner on an equivalent text and additionally pads a subgraph of the translation diagram to it.
It keeps track of all states visited so far together with the datum which caused the first transfer to the respective state.
This way we can reconstruct the last first-time-visited state while observing the equivalent text sequence.
Moreover, the equivalent text prevents the iterative learner from returning to a previously visited state but the last one and hence enables the required convergence.

However, if restrictions are not semantic, then iterative and bounded memory states learning can differ.
We show this concretely for the case of so-called \emph{strongly non-U-shaped} learning in Theorem~\ref{thm:itBMSDifferSNU}.
Inspired by cognitive science research \cite{Str-Sta:b:82}, \cite{Mar:b:92} a semantic version of this requirement was defined in \cite{Bal-Cas-Mer-Ste-Wie:j:08} and later the syntactic variant was introduced in \cite{Cas-Moe:j:11:optLan}.
Both requirements have been extensively studied, see \cite{Car-Cas:j:13:topics} for a survey and moreover \cite{Cas-Koe:j:11:memLess}, \cite{case2016strongly}, \cite{KSS17}.
The proof of Theorem~\ref{thm:itBMSDifferSNU} uses an intricate ORT-argument, which might suggest that the two settings, while different, are very similar nonetheless.
It is based on the proof that strong non-U-shapedness restricts $\BMS_\ast\Ex$-learning.
The proof of the latter result combines the techniques for showing that strong non-U-shapedness restricts iterative learning, as proved in \cite[Theorem 5.7]{Cas-Koe:j:11:memLess}, and that not every class strongly monotonically learnable by an iterative learner is strongly non-U-shapedly learnable by an iterative learner, see \cite[Theorem~5]{jain2016role}.
Moreover, it relies on showing that state decisiveness can be assumed in Lemma~\ref{StateDecisive}.

The remainder of Section~\ref{SyntacticLearning} completes the map given in Figure~\ref{fig:TxtBMSEx} for the case of syntactic restrictions (since these do not carry over from the setting of iterative learning).
All syntactic learning requirements are closely related to strongly locking learners.
The fundamental concept of a locking sequence was introduced by \cite{Blu-Blu:j:75}.
For a similar purpose than ours \cite{jain2016role} introduced strongly locking learners.
We generalize their construction for certain syntactically restricted iterative learners from a strongly locking iterative learner.
Finally, we obtain that all non-semantic learning restrictions also coincide for $\BMS_\ast$-learning.

\ignore{
We have the following \emph{known} results.
\begin{eqnarray}
\label{eq:NUBMS2} \forall c\in\N_{>0} \;\; [\BMS_c\Txt\Bc] & = & [\BMS_{(c+1)!}\Txt\Ex];\quad \cite[Th.~34]{Car-Cas-Jai-Ste:j:07}
\\
\label{eq:BMSBcEx} [\BMS_2\Txt\Ex] & = & [\BMS_2\Txt\NU\Ex];\quad \cite[Th.~35]{Car-Cas-Jai-Ste:j:07}\\
\label{eq:BMSHrrchy1} \forall c\in\N \;\; [\BMS_c\Txt\Ex] & \subsetneq & [\BMS_{c+1}\Txt\Ex];\quad \cite[Th.~37]{Car-Cas-Jai-Ste:j:07}\\
\label{eq:BMSHrrchy2} \bigcup_{c\in\N}[\BMS_c\Txt\Ex] & \subsetneq & [\BMS_\ast\Txt\Ex] \quad \subsetneq \quad [\Txt\BMS_\ast\Ex];\quad  \cite[Rem.~38]{Car-Cas-Jai-Ste:j:07}\\
\label{eq:ClassBMSStarIt} [\Txt\BMS_\ast\Ex] & = & [\It\Txt\Ex];\quad \cite[Rem.~38]{Car-Cas-Jai-Ste:j:07}\\
\label{eq:BMSStarItConf} [\BMS_\ast\Txt\Ex] & = & [\It\mathrm{Confident}\Txt\Ex];\quad \cite[Rem.~38]{Car-Cas-Jai-Ste:j:07}\\
\label{eq:NUBMS3} \bigcup_{n \in \natnum_{>0}} [\BMS_n\Txt\NU\Ex] & \subsetneq & [\BMS_3\Txt\Ex];\quad \cite[Th.~3.3]{Cas-Koe:j:11:memLess}\\
\label{eq:NUBMS3} [\BMS_\ast\Txt\NU\Ex] & = & [\BMS_\ast\Txt\Ex];\quad \cite[Cor.~3.6]{Cas-Koe:j:11:memLess}\\
\label{eq:Canny} [\mathbf{Can}\BMS_\ast/\It\Txt(\SNU)\Ex] & = & [\BMS_\ast/\It\Txt(\SNU)\Ex]; \quad \cite[Rem.~3.8]{Cas-Koe:j:11:memLess}\\
\label{eq:SNUBMS2} [\BMS_2\Txt\Ex] & \not\subseteq & [\BMS_\ast\Txt\SNU\Ex];\quad \cite[Th.~3.10]{Cas-Koe:j:11:memLess}\\
\label{eq:DecBMS2} [\BMS_2\Txt\Dec\Ex] & \subsetneq & [\BMS_2\Txt\Ex].\quad \cite[Rem.~3.12]{Cas-Koe:j:11:memLess}
\end{eqnarray}

We later on refer to the results:
\begin{eqnarray}
\label{eq:SNUIt} [\It\Txt\Ex] & \separate & [\It\Txt\SNU\Ex];\quad \cite{Cas-Koe:c:10:colt}\\
\label{eq:NUIt} [\It\Txt\Ex] & = & [\It\Txt\NU\Ex];\quad \cite{Cas-Moe:j:08:NUIt} \\
\label{eq:MonSNUIt} [\It\Txt\Mon\Ex] & \not\subseteq & [\It\Txt\SNU\Ex];
\quad \cite{jain2016role} \\
\label{eq:SMonSNUIt} [\It\Txt\SMon\Ex] & \not\subseteq & [\It\Txt\SNU\Ex];
\quad \cite{jain2016role}
\end{eqnarray}
}

%
%
%
%

\section{Learners, Success Criteria and other Terminology}
\label{sec:prelims}

As far as possible, we follow \cite{Jai-Osh-Roy-Sha:b:99:stl2} on the learning theoretic side and \cite{Odi:b:99} for computability theory. We recall the most essential notation and definitions.

\medskip
We let $\N$ denote the \emph{natural numbers} including $0$.
For a function $f$ we write $\dom(f)$ for its \emph{domain} and $\ran(f)$ for its \emph{range}.
If we deal with (a subset of) a cartesian product, we are going to refer to the \emph{projection functions} to the first or second coordinate by $\pr_1$ and $\pr_2$, respectively.

Further, $\Seq{X}$ denotes the \emph{finite sequence}s over the set $X$ and $\ISeq{X}$ stands for the \emph{countably infinite sequence}s over $X$.
For every $\sigma \in \Seq{X}$ and $t \leq |\sigma|$, $t \in \N$, we let $\sigma[t] := \{ (s,\sigma(s)) \mid s < t \}$ denote the \emph{restriction of $\sigma$ to $t$}.
Moreover, for sequences $\sigma, \tau \in \Seq{X}$ their concatenation is denoted by $\sigma\concat\tau$.
Finally, we write $\last(\sigma)$ for the last element of $\sigma$ , $\sigma(|\sigma|-1)$, and $\sigma^-$ for the initial segment of $\sigma$ without $\last(\sigma)$, i.e. $\sigma[|\sigma|-1]$. Clearly, $\sigma=\sigma^-\concat\last(\sigma)$.

For a finite set $D\subseteq\N$ and a finite sequence $\sigma\in\Seq{X}$, we denote by $\langle D\rangle$ and $\langle \sigma\rangle$ a canonical index for $D$ or $\sigma$, respectively. Further, we fix a Goedel pairing function $\langle.,.\rangle$ with two arguments.

\medskip
Let $L \subseteq \N$. We interpret every $n\in\N$ as a code for a word.
If $L$ is recursively enumerable, we call $L$ a \emph{language}.

\medskip
We fix a programming system $\varphi$ as introduced in \cite{Roy-Cas:b:94}. Briefly, in the $\varphi$-system, for a natural number $p$, we denote by $\varphi_p$ the partial computable function with program code $p$.
We call $p$ an \emph{index} for $W_p$ defined as $\dom(\varphi_p)$.

In reference to a Blum complexity measure $\Phi_p$, for all $p, t \in \N$, we denote by $W^t_p \subseteq W_p$ the recursive set of all natural numbers less or equal to $t$, on which the machine executing $p$ halts in at most $t$ steps, i.e.
$$W^t_p = \{ x \mid x\leq t \:\wedge\: \Phi_p(x)\leq t \}.$$
Moreover, the well-known s-m-n theorem gives finite and infinite recursion theorems, see \cite{case1974periodicity}, \cite{Cas:j:94:self}, \cite{Odi:b:99}.
We will refer to Case's Operator Recursion Theorem ORT in its 1-1-form, see for example \cite{Koe:th:09}.


\bigskip
Throughout the paper, we let $\Sigma=\N\cup\{\#\}$ be the input alphabet with $n\in\N$ interpreted as code for a word in the language and $\#$ interpreted as pause symbol, i.e. no new information.
Further, let $\Omega=\N\cup\{?\}$ be the output alphabet with $p\in\N$ interpreted as $\varphi$-index and $?$ as no hypothesis or repetition of the last hypothesis, if existent. A function with range $\Omega$ is called a hypothesis generating function.

A \emph{learner} is always a (partial) computable function $$M: \dom(M) \subseteq \Seq{\Sigma} \to \Omega.$$ The set of all total computable functions $M: \Seq{\Sigma} \to \Omega$ is denoted by $\totalCp$.

\medskip
Let $f \in \Seq{\Sigma} \cup \ISeq{\Sigma}$, then the \emph{content of $f$}, defined as $\cnt(f) := \ran(f)\setminus\{\#\}$, is the set of all natural numbers, about which $f$ gives some positive information. 
The \emph{set of all texts for the language $L$} is defined as
$$\Txt(L) := \{ T \in \ISeq{\Sigma} \mid \cnt(T) = L \}.$$

\begin{defn}
\label{def:it}
Let $M$ be a learner. $M$ is an \emph{iterative learner} or \emph{$\It$-learner}, for short $M \in \It$, if there is a computable (partial) hypothesis generating function 
$h_M: \Omega\times\Sigma \to \Omega$
such that $M=h_M^\ddagger$ where $h_M^\ddagger$ is defined on finite sequences by
\begin{align*}
h_M^\ddagger(\epsilon)&=\:?;\\
h_M^\ddagger(\sigma\concat x)&=h_M(h_M^\ddagger(\sigma),x).
\end{align*}
\end{defn}

\begin{defn}
\label{def:bms}
Let $M$ be a learner. $M$ is a \emph{bounded memory states learner}  or \emph{$\BMS$-learner}, for short $M \in \BMS$, if there are a computable (partial) hypothesis generating function 
$h_M:\N\times \Sigma \to \Omega$
and a computable (partial) state transition function
$s_M:\N\times \Sigma \to \N$
such that $\dom(h_M)=\dom(s_M)$ and $M=h_M^\ast$ where $h_M^\ast$ and $s_M^\ast$ are defined on finite sequences by
\begin{align*}
s_M^\ast(\epsilon)&=0; \\
h_M^\ast(\sigma\concat x)&=h_M(s_M^\ast(\sigma),x); \\
s_M^\ast(\sigma\concat x)&=s_M(s_M^\ast(\sigma),x).
\end{align*}
\end{defn}

Note that every iterative learner gives a $\BMS$-learner by identifying the hypothesis space $\Omega$ with the set of states via a computable bijection between $\N$ and $\Omega$.
The resulting $\BMS$-learner will succeed on the same languages the iterative learner does learn. Further, as the set of visited states contains exactly all hypotheses the learner puts out, this $\BMS$-learner only uses finitely many states on all texts for languages it explanatory learns.
In \cite[Rem.~38]{Car-Cas-Jai-Ste:j:07} it is claimed that $\BMS_\ast$-learners and iterative learners are equally powerful on texts.
This also follows from our more general Lemma~\ref{BMS-It_sem-afsoet}. The above intuition is formalized in the corresponding proof.

\medskip
Definition~\ref{def:bms} may be stated more generally for arbitrary finite or infinite sets of states $Q$, instead of $\N$. Moreover, $s_M^\ast$ and $h_M^\ast$ can easily be generalized to functions taking also a starting state $s$ as input by
\begin{align*}
s_M^\ast(s,\epsilon)&=s; \\
h_M^\ast(s,\sigma\concat x)&=h_M(s_M^\ast(s,\sigma),x); \\
s_M^\ast(s,\sigma\concat x)&=s_M(s_M^\ast(s,\sigma),x).
\end{align*}

\bigskip
We now clarify what we mean by successful learning.

\begin{defn}
\label{def:convergence}
Let $M$ be a learner and $\CalL$ a collection of languages.
\begin{enumerate}
\item Let $L \in \CalL$ be a language and $T \in \Txt(L)$ a text for $L$ presented to $M$.
\begin{enumerate}
\item We call $h=(h_t)_{t\in\N} \in \ISeq{\Omega}$, where $h_t := M(T[t])$ for all $t \in \N$, the \emph{learning sequence of $M$ on $T$}.
\item $M$ \emph{learns $L$ from $T$ in the limit}, for short $M$ $\Ex$-learns $L$ from $T$ or $\Ex(M,T)$, if there exists $t_0\in\N$ such that $W_{h_{t_0}}=\cnt(T)$ and $\forall t \geq t_0 \; \left(\; h_t\neq\,? \; \Rightarrow \; h_t=h_{t_0}\;\right)$.
\end{enumerate}
\item $M$ \emph{learns $\CalL$ in the limit}, for short $M$ $\Ex$-learns $\CalL$, if $\Ex(M,T)$ for every $L\in\CalL$ and every $T \in \Txt(L)$.
\end{enumerate}
\end{defn}

\begin{defn}
Let $\CalL$ be a collection of languages.
$\CalL$ is \emph{learnable in the limit} or \emph{$\Ex$-learnable}, if there exists a learner $M$ that $\Ex$-learns $\CalL$.
\end{defn}

$\Ex$-learning is the most common definition for successful learning in inductive inference and corresponds to the notion of identifiability in the limit by \cite{Gol:j:67}, where the learner eventually decides on one correct hypotheses.

\bigskip
In our investigations, the most important additional requirement on a successful learning process for a $\BMS$-learner is to use finitely many states only, as stated in the following definition.

\begin{defn}
Let $M$ be a $\BMS$-learner and $T\in\Txt$. We say that $M$ \emph{uses finitely many memory states on $T$}, for short $\BMS_\ast(M,T)$, if $\{\, s_M^\ast(T[t]) \mid t \in \N \,\}$ is finite.
\end{defn}

We list the most common additional requirements regarding the learning sequence, which may tag a learning process. For this we first recall the notion of consistency of a sequence with a set.

\begin{defn}
\label{def:consfA}
Let $f \in \Seq{\Sigma}\cup\ISeq{\Sigma}$ and $A \subseteq \Sigma$. We define
\begin{align*}
\Comp(f,A) \quad &:\Leftrightarrow \quad \cnt(f) \subseteq A
\end{align*}
and say \emph{$f$ is consistent with $A$}.
\end{defn}

The listed properties of the learning sequence have been at the center of different investigations. Studying how they relate to one another did begin in \cite{kotzing2016map}, \cite{kotzing2016towards}, \cite{jain2016role} and \cite{As-Koe-Sei2018_informants}.

\begin{defn}
\label{def:LearningRestrictions}
Let $M$ be a learner, $T \in \Txt$ and $h=(h_t)_{t\in\N} \in \ISeq{\Omega}$ the learning sequence of $M$ on $T$, i.e. $h_t=M(T[t])$ for all $t\in\N$..
We write
\begin{enumerate}
        \item $\Conv(M,T)$ (\cite{angluin1980inductive}), if $M$ is \emph{conservative on $T$}, i.e., for all $s, t$ with $s \leq t$ holds \\
        $\Comp(T[t],W_{h_s}) \;\Rightarrow\; h_s = h_t.$
        \item  $\Dec(M,T)$ (\cite{Osh-Sto-Wei:j:82:strategies}),
        if $M$ is \emph{decisive on $T$}, i.e.,
        for all $r, s, t$ with $r \leq s \leq t$ holds \\
        $W_{h_r} = W_{h_t} \;\Rightarrow\; W_{h_r} = W_{h_s}.$
        \item $\Caut(M,T)$ (\cite{STL1}),
		if $M$ is \emph{cautious on $T$}, i.e.,
		for all $s, t$ with $s \leq t$ holds
        $\neg W_{h_t} \subsetneq W_{h_s}.$
		\item $\WMon(M,T)$ (\cite{j-mniifp-91},\cite{Wie:c:91}), if $M$ is \emph{weakly monotonic on $T$}, i.e., for all $s, t$ with $s \leq t$ holds
        $\Comp(T[t], W_{h_s})
        \;\Rightarrow\; W_{h_s} \subseteq W_{h_t}.$
        \item $\Mon(M,T)$ (\cite{j-mniifp-91},\cite{Wie:c:91}), if $M$ is \emph{monotonic on $T$}, i.e., for all $s, t$ with $s \leq t$ holds
        $W_{h_s} \cap \cnt(T) \subseteq W_{h_t}\cap\ps(T).$
        \item  $\SMon(M,T)$ (\cite{j-mniifp-91},\cite{Wie:c:91}), if $M$ is \emph{strongly monotonic on $T$}, i.e., for all $s, t$ with $s \leq t$ holds
        $W_{h_s} \subseteq W_{h_t}.$
        \item  $\NU(M,T)$ (\cite{Bal-Cas-Mer-Ste-Wie:j:08}), if $M$ is \emph{non-U-shaped on $T$}, i.e., for all $r, s, t$ with $r \leq s \leq t$ holds \linebreak
        $W_{h_r} = W_{h_t} = \cnt(T) \;\Rightarrow\; W_{h_r} = W_{h_s}.$
        \item  $\SNU(M,T)$ (\cite{Cas-Moe:j:11:optLan}), if $M$ is \emph{strongly non-U-shaped on $T$}, i.e., for all $r, s, t$ with $r \leq s \leq t$ holds
        $W_{h_r} = W_{h_t} = \cnt(T) \;\Rightarrow\; h_r = h_s.$
        \item  $\SDec(M,T)$ (\cite{kotzing2016map}), if $M$ is \emph{strongly decisive on $T$}, i.e., for all $r, s, t$ with $r \leq s \leq t$ holds
        $W_{h_r} = W_{h_t} \;\Rightarrow\; h_r = h_s.$
				\item  $\Wb(M,T)$ (\cite{kotzing2016towards}), if $M$ is \emph{witness-based on $T$}, i.e., for all $r, t$ such that for some $s$ with $r < s \leq t$
				holds $h_r\neq h_s$ we have
        $\cnt(T[s])\cap (W_{h_t}\setminus W_{h_r}) \neq\varnothing.$
\end{enumerate}%
\end{defn}


It is easy to see that $\Conv(M,T)$ implies $\SNU(M,T)$ and $\WMon(M,T)$;
$\SDec(M,T)$ implies \linebreak $\Dec(M,T)$ and $\SNU(M,T)$;
$\SMon(M,T)$ implies $\Caut(M,T), \Dec(M,T), \Mon(M,T)$, $\WMon(M,T)$
and finally $\Dec(M,T)$, $\WMon(M,T)$ and $\SNU(M,T)$ imply $\NU(M,T)$.
Figure~\ref{fig:TxtBMSEx} includes the resulting backbone with arrows indicating the aforementioned implications.
Further, $\Wb(M,T)$ implies $\Conv(M,T)$, $\SDec(M,T)$ and $\Caut(M,T)$.

In order to characterize what successful learning means, these predicates may be combined with the explanatory convergence criterion. For this, we let
$\Delta := \{ \,\Caut, \Conv, \Dec, \SDec, \WMon, \Mon,\SMon,$ $\NU,\SNU, \mathbf{T} \,\}$ denote the set of \emph{admissible learning restrictions}, with $\mathbf{T}$ standing for no restriction.
Further, a \emph{learning success criterion} is a predicate being the intersection of the convergence criterion $\Ex$ with arbitrarily many admissible learning restrictions.
This means that the sequence of hypotheses has to converge and in addition has the desired properties.
Therefore, the collection of all learning success criteria is
$$\{ \: \bigcap_{i=0}^n \delta_i \cap \Ex \mid n \in \N, \forall i \leq n (\delta_i \in \Delta) \}.$$

Note that plain explanatory convergence is a learning success criterion by letting $n=0$ and $\delta_0=\mathbf{T}$.

\smallskip
We refer to all $\delta \in \{\Caut,\Cons,\Dec,\Mon,\SMon,\WMon,\NU,\mathbf{T}\}$ also as \emph{semantic} learning restrictions, as they do not require the learner to settle on exactly one hypothesis.

In order to state observations about how two ways of defining learning success relate to each other, the learning power of the different settings is encapsulated in notions $[\alpha\Txt\beta]$ defined as follows.

\begin{defn}
\label{def:learncrit}
Let $\alpha$ be a property of partial computable functions from the set $\Seq{\Sigma}$ to $\N$ and $\beta$ a learning success criterion.
We denote by $[\alpha\Txt\beta]$ the set of all collections of languages that are $\beta$-learnable from texts by a learner $M$ with the property $\alpha$.
\end{defn}

At position $\alpha$, we restrict the set of admissible learners for example by requiring them to be iterative or finite bounded memory states learners. The properties stated at position $\alpha$ are \emph{independent of learning success}.
In contrast, at position $\beta$, the required learning behavior and convergence criterion are specified.
We do not use separators in the notation to stay consistent with established notation in the field that was inspired by \cite{Jai-Osh-Roy-Sha:b:99:stl2}.

\medskip
For example, a collection of languages $\CalL$ lies in $[\BMS\Txt\BMS_\ast\Conv\Ex]$ if and only if there is a bounded memory states learner $M$ conservatively explanatory learning every $L \in \CalL$ from texts while using only finite memory.
More concretely, for all $L \in\CalL$ and for every text $T\in\Txt(L)$ we have $\Conv(M,T)$, $\BMS_\ast(M,T)$ and $\Ex(M,T)$.

\medskip
The proofs of Lemmata~\ref{BMS-It_sem-afsoet} and \ref{StateDecisive} employ the following property of learning requirements and learning success criteria, that applies to all such considered in this paper.

\begin{defn}
\label{def:afsoet}
Denote the set of all unbounded and non-decreasing functions by $\Simu$, i.e., $$\Simu := \{ \,\simu: \N \to \N \mid \forall x \in \N \,\exists t \in \N \colon \simu(t) \geq x \text{ and } \forall t \in \N \colon \simu(t+1) \geq \simu(t) \,\}.$$
Then every $\simu \in \Simu$ is a so called \emph{admissible simulating function}.

\smallskip
A predicate $\beta$ on pairs of learners and texts \emph{allows for simulation on equivalent text}, if for all simulating functions $\simu \in \Simu$, all texts $T, T' \in \Txt$ and all learners $M, M'$ holds:
Whenever we have $\cnt(T'[t]) = \cnt(T[\simu(t)])$ and $M'(T'[t]) = M(T[\simu(t)])$ for all $t \in \N$, from $\beta(M,T)$ we can conclude $\beta(M', T')$.
\end{defn}

Intuitively, as long as the learner $M'$ conjectures $h'_t=h_{\simu(t)}=M(T[\simu(t)])$ at time $t$ and has, in form of $T'[t]$, the same data available as was used by $M$ for this hypothesis, $M'$ on $T'$ is considered to be a simulation of $M$ on $T$.

\bigskip
It is easy to see that all learning success criteria considered in this paper allow for simulation on equivalent text.

\section{Relations between Semantic Learning Requirements}

We show that bounded memory states learners and iterative learners have equal learning power, when a semantic learning requirement is added to the standard convergence criterion. With this the results from iterative learning are transferred to this setting.

\bigskip
The following lemma formally establishes the equal learning power of iterative and $\BMS_\ast$-learning for all learning success criteria but $\Conv$, $\SDec$ and $\SNU$. We are going to prove in Section~\ref{SyntacticLearning} that even for the three aforementioned non-semantic additional requirements we obtain the same behavior.

\begin{lem}
\label{BMS-It_sem-afsoet}
Let $\delta$ allow for simulation on equivalent text.
\begin{enumerate}
\item \label{BMSIt} We have $[\Txt\BMS_\ast\delta\Ex]\supseteq[\It\Txt\delta\Ex].$
\item \label{second} If $\delta$ is semantic then $[\Txt\BMS_\ast\delta\Ex]=[\It\Txt\delta\Ex].$
\end{enumerate}
\end{lem}
\begin{proof}
While \ref{BMSIt} and ``$\supseteq$'' in \ref{second} are easy to verify by using the hypotheses as states, the other inclusion in \ref{second} is more challenging.
The iterative learner constructed from the $\BMS$-learner $M$ uses the hypotheses of $M$ on an equivalent text and additionally pads a subgraph of the translation diagram of $M$ to it.

(1) and ``$\supseteq$'' of (2).
Let $M$ be an iterative learner, i.e. there is a computable function $h_M: \Omega\times\Sigma\to\Omega$ with $M=h_M^\ddagger$ where $h_M^\ddagger(\epsilon)=\:?$ and $h_M^\ddagger(\sigma\concat x)=h_M(h_M^\ddagger(\sigma),x)$ for all $\sigma\in\Sigma^{<\omega}$ and $x\in\Sigma$.
We show that $M$ can be obtained as a state driven learner by using the hypotheses also as states. For this, we fix the computable bijection $\pi: Q \to \Omega$ with computable inverse, defined by $\pi(0)=\:?$ and $\pi(i)=i-1$ for all $i>0$. Then the learner $N=h^\ast_N$ with $\langle s_N,h_N \rangle(q,x)=(\pi^{-1}(h_M(\pi(q),x)),h_M(\pi(q),x))$ is as wished because the state corresponds via $\pi$ directly to the last hypothesis of $M$ and so the learners $M$ and $N$ act identically.

Formally, this follows by an induction showing for every $\tau\in\Sigma^{<\omega}$ that $s_N^\ast(\tau)=\pi^{-1}(M(\tau))$ and moreover if $|\tau|>0$ we have $N(\tau)=M(\tau)$.
The claim holds for $\tau=\epsilon$, because of
$s_N^\ast(\epsilon)=0=\pi^{-1}(M(\epsilon))$.
In case there are $\sigma\in\Sigma^{<\omega}$ and $x\in\Sigma$ such that $\tau=\sigma\concat x$, we may assume $s_N^\ast(\sigma)=\pi^{-1}(M(\sigma))$ and obtain
\begin{align*}
s_N^\ast(\tau) &\stackrel{\mathrm{Def.}\, s_N^\ast}{=} s_N(s_N^\ast(\sigma),x)\stackrel{s_N^\ast(\sigma)=\pi^{-1}(M(\sigma))}{=} s_N(\pi^{-1}(M(\sigma)),x)
\stackrel{\mathrm{Def.}\, s_N}{=} \pi^{-1}(h_M(M(\sigma),x))
\stackrel{M = h_M^\ddagger}{=} \pi^{-1}(M(\tau)),\\
N(\tau) &\stackrel{N=h_N^\ast}{=} h_N(s_N^\ast(\sigma),x)
\stackrel{s_N^\ast(\sigma)=\pi^{-1}(M(\sigma))}{=} h_N(\pi^{-1}(M(\sigma)),x)
\stackrel{\mathrm{Def.}\, h_N}{=} h_M(M(\sigma),x))
\stackrel{M = h_M^\ddagger}{=} M(\tau). 
\end{align*}
That $M$ in case of learning success uses only finitely many states follows immediately from the $\Ex$-convergence, implying to output only finitely many pairwise distinct hypotheses.

\medskip
``$\subseteq$'' of (2).
Let $\CalL \in [\Txt\BMS_\ast\delta\Ex]$ be witnessed by the learner $M$, i.e., there is $\langle s_M, h_M \rangle: Q\times\Sigma\to Q\times\Omega$ such that $M=h_M^\ast$. 
Further, we may assume that for all $L\in\CalL$ and $T\in\Txt(L)$ the set of visited states $s_M^\ast[\{T[t]\mid t\in\N\}]$ is finite and $M$ $\delta\Ex$-learns $L$ from $T$.

Intuitively, the iterative learner $M_\It$ uses the hypotheses of $M$ on an equivalent text $\hat{T}$ and additionally pads a subgraph $\visit(\sigma)$ of the translation diagram of the $\BMS$-learner $M$ to it. In $\visit(\sigma)$, which is being build after having observed $\sigma$, we keep track of all states visited so far together with the datum which caused the first transfer to the respective state. In order to assure $\Ex$-convergence, we do not change the subgraph in case the new state had already been visited after some proper initial segment of $\sigma$ was observed. From $\visit(\sigma)$ we can reconstruct the last first-time-visited state $s^\ast_{M_\It}(\sigma)$ of $M$ while observing the equivalent sequence corresponding to $\sigma$. Moreover, we build the equivalent text $\hat{T}$ by inserting a path of already observed data leading to state $s^\ast_{M_\It}(\sigma)$, in case this is necessary to prevent the learner $M_\It$ from returning to a previously visited state but the last one. With this strategy we make sure that the last state is the one we are currently in, as keeping track of the current state while observing the original text may destroy the $\Ex$-convergence.

Formally, we define functions $\pump: \Seq{\Sigma}\setminus\{\epsilon\} \times\N \to \Seq{\Sigma}$ and $\visit: \Seq{\Sigma} \to \Seq{\Sigma}$ by
\begin{align*}
\pump(\visit(\sigma),x) &= \begin{cases}
x, &\text{if } s_M(s_{M_\It}^\ast(\sigma),x)\notin \states[\visit(\sigma)]; \\
x\concat \Mypath(s_M(s_{M_\It}^\ast(\sigma),x),s^\ast_{M_\It}(\sigma)), &\text{otherwise;}
\end{cases} \\
\visit(\epsilon) &= \epsilon;\\
\visit(\sigma\concat x) &= \begin{cases}
\visit(\sigma)\concat \langle s_M(s_{M_\It}^\ast(\sigma),x),x \rangle, &\text{if } s_M(s_{M_\It}^\ast(\sigma),x)\notin \states[\visit(\sigma)]; \\
\visit(\sigma), &\text{otherwise;}
\end{cases}
\end{align*}
with the application of the projection to the first coordinate extracting the set of visited states.
Moreover, for states $s_0,s_1\in S$ with $\Mypath(s_0,s_1)$ we refer to the unique sequence $(\sigma(i), \sigma(i+1), \ldots, \sigma(j))$ of second coordinates in $\visit(\sigma)$ such that $(s_0,\sigma(i))\concat\ldots\concat(s_1,\sigma(j))$ is an intermediate sequence in $\visit(\sigma)$.
The learner $M_\It$ is now defined by
$$M_\It(\sigma\concat x)=\pad(h_M^\ast(s^\ast_{M_\It}(\sigma),\pump(\visit(\sigma),x)),\visit(\sigma\concat x)).$$
By construction
$s^\ast_{M_\It}(\sigma)=\last(\states(\visit(\sigma)))$ 
and therefore the hypothesis of $M_\It$ on some sequence $\sigma\concat x$ is always only based on $\visit(\sigma)$ and $x$, which makes $M_\It$ iterative.

The text $\hat{T}=\bigcup_{t\in\N}\tau_t$ with $\tau_0=\epsilon$ and $\tau_{t+1}=\tau_t\concat\pump(\visit(T[t]),T(t))$ is a text for $L$. Let $\simu: \N\to\N, t \mapsto |\tau_t|$ be the corresponding simulating function.
As for all $t\in\N$ holds $\cnt(T[t])=\cnt(\hat{T}[\simu(t)])$ and $M_\It(T[t])=\pad(M(\hat{T}[\simu(t)]),\visit(T[t]))$, we obtain $W_{M_\It(T[t])}=W_{M(\hat{T}[\simu(t)])}$ and because $\delta$ is semantic and afsoet, we conclude the semantic $\delta$-convergence of $M_\It$ on $T$.
Having in mind that $M$ uses only finitely many pairwise distinct states $\visit(T[t])$ stabilizes.
Paired with the $\Ex$-convergence of $M$ on $\hat{T}$ we conclude the $\Ex$-convergence of $M_\It$ on $T$.
\end{proof}

Note that obviously the proof is identical for learning from positive and negative information, introduced by \cite{Gol:j:67}.
In this learning model the information the learner receives is labeled, like in binary classification, and has to be complete in the limit.
See \cite{As-Koe-Sei2018_informants} for a formal definition, a summary of results on this model and the complete map.

\bigskip
With Lemma~\ref{BMS-It_sem-afsoet} the following results transfer from learning with iterative learners and it remains to investigate the relations to and between the non-semantic requirements $\Conv, \SDec$ and $\SNU$. 

\begin{thm}
\begin{enumerate}
\item $[\Txt\BMS_\ast\NU\Ex]=[\Txt\BMS_\ast\Ex]$
\item $[\Txt\BMS_\ast\Dec\Ex]=[\Txt\BMS_\ast\WMon\Ex]=[\Txt\BMS_\ast\Caut\Ex]=[\Txt\BMS_\ast\Ex]$
\item $[\Txt\BMS_\ast\Mon\Ex]\subsetneq[\Txt\BMS_\ast\Ex]$
\item $[\Txt\BMS_\ast\SMon\Ex]\subsetneq[\Txt\BMS_\ast\Mon\Ex]$
\end{enumerate}
\end{thm}
\begin{proof}
The respective results for iterative learners can be found in
\cite[Theorem~2]{Cas-Moe:j:08:NUIt},
\cite[Theorem~10]{jain2016role},
\cite[Theorem~3]{jain2016role}
and \cite[Theorem~2]{jain2016role}.
\end{proof}

\section{Relations to and between Syntactic Learning Requirements}
\label{SyntacticLearning}

The following lemma establishes that we may assume $\BMS_\ast$-learners to never go back to withdrawn states.
This is essential in almost all of the following proofs.
It can also be used to simplify the proof of Lemma~\ref{BMS-It_sem-afsoet}.

\begin{lem}
\label{StateDecisive}
Let $\beta$ be a learning success criterion allowing for simulation on equivalent text and $\CalL \in [\Txt\BMS_\ast\beta]$.
Then there is a $\BMS$-learner $N$ such that $N$ never returns to a withdrawn state and $\BMS_\ast\beta$-learns $\CalL$ from texts.
\end{lem}
\begin{proof}
Let $M$ be a $\BMS$-learner with $\CalL \in \Txt\BMS_\ast\beta(M)$.
We employ a construction similar to the one in the proof of Theorem~\ref{BMS-It_sem-afsoet}.
Again for $\visit\in\Seq(Q\times\Sigma)$ with pairwise distinct first coordinates and $s'\in\pr_1[\visit]$ by $\Mypath(\visit,s')$ we denote the unique sequence of second coordinates $x_0\concat\ldots\concat x_\xi$ of $\visit$ such that $(s',x_0)\concat\ldots\concat(\last(\pr_1[\visit]),x_\xi)$ is a final segment of $\visit$.
The $\BMS$ learner $N$ is initialized with state $\pad(0,(0,\#))$ and for every $s\in Q$, $\visit\in\Seq(Q\times\Sigma)$ and $x\in\Sigma$ defined by 
\begin{align*}
s_N(\langle s,\visit\rangle,x)&= \begin{cases}
\langle s,\visit \rangle, &\text{if } s_M(s,x)\in \pr_1[\visit]; \\
\langle s_M(s,x),\visit\concat(s_M(s,x),x)\rangle, &\text{otherwise;}
\end{cases} \\
h_N(\langle s,\visit\rangle,x)&=\begin{cases}
h_M^\ast(s,x\concat\Mypath(\visit,s_M(s,x))), &\text{if } s_M(s,x)\in\pr_1[\visit]; \\
h_M(s,x), &\text{otherwise.}
\end{cases}
\end{align*}
By construction $N$ is a $\BMS_\ast$-learner, as it only uses states $\langle s,\visit\rangle$ where $s=\pr_1(\last(\visit))$ is a state used by $M$ and for every $s\in Q$, visited by $M$, there is exactly one sequence $\visit\in\Seq(Q\times\Sigma)$ such that $\langle s,\visit \rangle$ is used by $N$.
The learner $N$ simulates $M$ on an equivalent text just as in the proof of Theorem~\ref{BMS-It_sem-afsoet}.
\end{proof}

%
%

We show that strongly monotonically $\BMS_\ast$-learnability does not imply strongly non-U-shapedly $\BMS_\ast$-learnability.

\begin{thm}
\label{BMS-SMonSNU}
$[\Txt\BMS_\ast\SMon\Ex] \not\subseteq [\Txt\BMS_\ast\SNU\Ex]$
\end{thm}
\begin{proof}
We define a self-learning $\BMS$-learner $M$ and with a tailored ORT-argument there can not be a $\BMS$-learner strongly non-U-shapedly learning all languages that $M$ learns strongly monotonically.

Consider the $\BMS$-learner $M$ initialized with state $\langle \,?,\langle\varnothing\rangle \rangle$ and $h_M$ and $s_M$ for every $e\in\Omega$, $D\subseteq\N$ finite and $x\in\Sigma$ defined by:
\begin{align*}
s_M(\langle e,\langle D \rangle \rangle,x)&=\begin{cases}
\langle e,\langle D\rangle \rangle, &\text{if } x\in D\cup\{\#\} \;\vee\; \varphi_x(e)= e;\\
\langle \varphi_x(e),\langle D\cup\{x\}\rangle\rangle, &\text{else if } \varphi_x(e)\neq e;\\
\uparrow, &\text{otherwise.}
\end{cases} \\
h_M(\langle e,\langle D \rangle\rangle,x)&= \begin{cases}
e, &\text{if } x\in D\cup\{\#\} \;\vee\; \varphi_x(e)= e;\\
\varphi_x(e), &\text{else if } \varphi_x(e)\neq e;\\
\uparrow, &\text{otherwise.}
\end{cases}
\end{align*}
Thus, $M$ is self-learning by interpreting the datum $x$ as a program and the conjectures are generated by applying this program to the last hypothesis. (We identify $\varphi_x$ with the function obtained by using a bijection from $\N$ to $\Omega$.)
Further, in form of the states, the last hypothesis as well as exactly the data that already lead to a mind change of $M$ is stored.

Let $\CalL=\Txt\BMS_\ast\SMon\Ex(M)$. 

Assume there is a $\BMS_\ast$-learner $N$ with hypothesis generating function $h_N$ and state transition function $s_N$, such that $\CalL\subseteq\Txt\BMS_\ast\SNU\Ex(N)$. By Lemma~\ref{StateDecisive} we assume that $N$ does not return to withdrawn states.

We are going to obtain a language $L\in \CalL$ not strongly non-U-shapedly learned by $N$ by applying 1-1 ORT and thereby refering to the c.e.~predicates $\MC$ and $\NoMC$ defined for fixed $a,b\in\totalCp$, all $k\in\N$ and $\sigma\in\Seq{\Sigma}$ with the help of the formulas $\psi_{k}(\ell)$, expressing that the $\BMS_\ast$-learner $N$ does not perform a mind- or state-change on the text $a[k]\concat b(k)\concat\#^\infty$ after having observed $a[k]\concat b(k)\concat\#^{\ell}$. The predicates state that $N$ does converge and (not) make a mind-change when observing $\sigma$ after having observed $a[k]\concat a(k)\concat \#^{\ell_k}$, with $\ell_k$ being the least $\ell$ with $\psi_k(\ell)$.
\begin{align*}
\psi_{k}(\ell) &\Leftrightarrow N(a[k]\concat b(k)\concat \#^{\ell})= N(a[k]\concat b(k)\concat \#^{\ell+1}) \;\wedge\; s^\ast_N(a[k]\concat b(k)\concat \#^{\ell})= s^\ast_N(a[k]\concat b(k)\concat \#^{\ell+1});\\
\NoMC(k,\sigma) &\Leftrightarrow \exists {\ell_k}\in\N\:(\,\psi_{k}(\ell_k) \:\wedge\: \forall \ell<\ell_k \:\neg \psi_k(\ell) \:\wedge\: N(a[k]\concat b(k)\concat \#^{{\ell_k}}\concat \sigma)\!\downarrow \;= N(a[k]\concat b(k)\concat \#^{{\ell_k}})\,); \\
\MC(k,\sigma) &\Leftrightarrow \exists {\ell_k}\in\N\:(\,\psi_{k}(\ell_k) \:\wedge\: \forall \ell<\ell_k \:\neg \psi_k(\ell) \:\wedge\: N(a[k]\concat b(k)\concat \#^{{\ell_k}}\concat \sigma)\!\downarrow \;\neq N(a[k]\concat b(k)\concat \#^{{\ell_k}}) \,).
\end{align*}
Now, let $p$ be an index for the program which on inputs $k\in\N$ and $\sigma\in\Seq{\Sigma}$ searches for $\ell_k$. In case $\ell_k$ exists, the program encoded in $p$ runs $N$ on $a[k]\concat b(k) \concat \#^{\ell_k} \concat \sigma$. Hence, $\Phi_p(k,\sigma)$ stands for the number of computation steps the program just described needs on input $k,\sigma$. By the definition of $p$ we have $\Phi_p(k,\sigma)\!\uparrow$ if and only if $\ell_k\!\uparrow$ or $N(a[k]\concat b(k) \concat \#^{\ell_k} \concat \sigma)\!\uparrow$. 

We abbreviate with $\Seq(a,i)=\Seq{}_{\leq i}(\ran(a[i])\cup\{\#\})$ the set of all finite sequences over $\ran(a[i])\cup\{\#\}$ with length at most $i$. Moreover, we employ a well-order $<_{a}$ on $\Seq(\ran(a))$ by letting $\rho <_a \sigma$ if and only if for the unique $i_\rho$ such that $\rho\in\Seq(a,i_\rho+1)\setminus\Seq(a,i_\rho)$ holds $\sigma\notin\Seq(a,i_\rho+1)$ or else $\sigma\notin\Seq(a,i_\rho)$ and at the same time $\langle\rho\rangle<\langle\sigma\rangle$.
For constructing $L$ we will also make use of the c.e.~sets
$$E_k=\{\, a(i) \mid \forall \sigma\in\Seq(a,i) \:\NoMC(k,\sigma) \;\vee\; (\,\exists \sigma \forall \rho<_a\sigma \:\NoMC(k,\rho) 
\:\wedge\: \Phi_p(k,\sigma)> i\,)\,\}.$$
It is easy to see that $E_k$ is finite and equals $\{\,a(i)\mid i < \max(\{i_{\sigma_0}\}\cup\{ \Phi_p(k,\sigma) \mid \sigma\leq_a\sigma_0\})\,\}$ if and only if for $\sigma_0\in\Seq(\ran(a))$ holds $\MC(k,\sigma_{0})$ and $\NoMC(k,\sigma)$ for all $\sigma<_a\sigma_{0}$. 
Otherwise $E_k=\ran(a)$.

By 1-1 ORT there are $a,b,e_1,e_2\in\totalCp$ with pairwise disjoint ranges and $e_0\in\N$, such that
\begin{align*}
\varphi_{a(i)}(e)&=\begin{cases}
e_0 &\text{if } e\in\{?,e_0\};\\
e_2(k) &\text{else if } e=e_1(k) \text{ for some } k\leq i;\\
e, &\text{otherwise;}
\end{cases} \\
\varphi_{b(k)}(e)&=\begin{cases}
e_1(k) &\text{if } e\in\{?,e_0\};\\
e, &\text{otherwise;}\\
\end{cases}\\
W_{e_0}&=\begin{cases}
\ran(a[t_0]) &\text{if $t_0$ is minimal with } \forall t\geq t_0\,(\, N(a[t])=N(a[t_0]) \wedge s_N^\ast(a[t])=s_N^\ast(a[t_0])\,); \\
\ran(a), &\text{no such $t_0$ exists.};
\end{cases}\\
W_{e_1(k)}&=
\cnt(a[k])\cup\{b(k)\}\cup\begin{cases}
E_k &\text{if } \exists \sigma_0 \,(\,\MC(k,\sigma_{0}) \,\wedge\, \forall \sigma<_a\sigma_{0}\, \NoMC(k,\sigma)\,);\\
\varnothing, &\text{otherwise;}
\end{cases} \\
W_{e_2(k)}&=
\cnt(a[k])\cup\{b(k)\}\cup E_k.
\end{align*}
As $W_{e_0}\in\CalL$ by construction, $N$ has to learn it and hence $t_0$ exists.

We first observe that there exists $\sigma_0$ such that $\MC(t_0,\sigma_{0})$ and $\NoMC(t_0,\sigma)$ for all $\sigma<_a\sigma_{0}$.
Assume otherwise, then either $\ell_{t_0}\!\!\uparrow$ or for all $\sigma\in\Seq(\ran(a))$ holds $\NoMC(t_0,\sigma)$ or for $\sigma_0$ minimal with $\neg \NoMC(t_0,\sigma_{0})$ we have $N(a[t_0]\concat b(t_0)\concat \#^{{\ell_{t_0}}}\concat \sigma_{0})\!\uparrow$.
Anyhow, this would mean $E_{t_0}=\ran(a)$.
By the definition of $e_1$, $e_2$ and our converse assumption we obtain $W_{e_1(t_0)}=\cnt(a[t_0])\cup\{b(t_0)\}$ and $W_{e_2(t_0)}=\ran(a)\cup\{b(t_0)\}$.
It can be easily checked that $W_{e_1(t_0)}$ and $W_{e_2(t_0)}$ are strongly monotonically learned by $M$ and hence lie in $\CalL$. As $N$ has to learn $W_{e_1(t_0)}$ from the text $a[t_0]\concat b(t_0)\concat \#^\infty$, we know $\ell_{t_0}\!\!\downarrow$ and moreover $W_{N(a[t_0]\concat b(t_0)\concat \#^\ell)}=W_{e_1(t_0)}$ holds for all $\ell\geq \ell_{t_0}$.
Moreover, $N$ has to learn $W_{e_2(t_0)}$ from all the texts $a[t_0]\concat b(t_0)\concat \#^{\ell_{t_0}}\concat \sigma\concat a$ with $\sigma\in\Seq(\ran(a))$.
Thus, $N(a[t_0]\concat b(t_0)\concat \#^{\ell_{t_0}}\concat \sigma)\!\downarrow$ for all $\sigma\in\Seq(\ran(a))$.
Because of our converse assumption, the only option left is $\NoMC(t_0,\sigma)$ for all $\sigma\in\Seq(\ran(a))$.
Since this is equivalent to $N(a[t_0]\concat b(t_0)\concat \#^{\ell_{t_0}} \concat \sigma)=N(a[t_0]\concat b(t_0)\concat \#^{\ell_{t_0}})$ 
for all $\sigma\in\Seq(\ran(a))$, $N$ cannot learn both $W_{e_1(t_0)}$ and $W_{e_2(t_0)}$. Hence $\sigma_0$ exists.

By the choice of $t_0$ and $\sigma_0$ we obtain $E_{t_0}=\cnt(a[t_1])$ for $t_1=\max(\{i_{\sigma_0}\}\cup\{ \Phi_p(k,\sigma) \mid \sigma\leq_a\sigma_0 \})\in\N$. Let $\hat{t}=\max\{t_0,t_1\}$ and $L=\cnt(a[\hat{t}])\cup\{b(t_0)\}$.  Then $W_{e_1(t_0)}=W_{e_2(t_0)}=L \in \CalL$ and by construction of $E_{t_0}$ we have $\Cons(\sigma_{0},L)$.
Because of $\hat{t}\geq t_0$, we obtain $s_N^\ast(a[\hat{t}])=s_N^\ast(a[t_0])$. With this and the choice of $t_0$ we conclude
$N(a[\hat{t}]\concat b(t_0)\concat \#^\ell)=N(a[t_0]\concat b(t_0)\concat \#^\ell)$ 
for all $\ell\in\N$.
Further, as $N$ learns $L$ from the text $a[\hat{t}]\concat b(t_0)\concat \#^\infty$ we have $W_{N(a[\hat{t}]\concat b(t_0)\concat \#^{\ell_{t_0}})}=L$.
On the other hand by $\MC(t_0,\sigma_{0})$ we obtain 
$N(a[\hat{t}]\concat b(t_0)\concat \#^{\ell_{t_0}})\neq N(a[\hat{t}]\concat b(t_0)\concat \#^{\ell_{t_0}}\concat \sigma_{0})$, which forces $N$ to perform a syntactic U-shape on the text $a[\hat{t}]\concat b(t_0)\concat \#^{\ell_{t_0}}\concat \sigma_{0}\concat \#^\infty$ for $L$.
\end{proof}

\bigskip
For inferring the relations between the syntactic learning requirements $\SNU$, $\SDec$ and $\Conv$, we refer to $\Wb$.
All these criteria are closely related to strongly locking learners, which we define in the following.

It was observed by \cite{Blu-Blu:j:75} that the learnability of every language $L$ by a learner $M$ is witnessed by a sequence $\sigma$, consistent with $L$, such that $M(\sigma)$ is an index for $L$ and no extension of $\sigma$ consistent with $L$ will lead to a mind-change of $M$.
Such a sequence $\sigma$ is called \emph{(sink-)locking sequence for $M$ on $L$}.
For a similar purpose as ours \cite{jain2016role} introduced strongly locking learners.
A learner $M$ acts strongly locking on a language $L$, if for every text $T$ for $L$ there is an initial segment $\sigma$ of $T$ that is a locking sequence for $M$ on $L$.

The proof of the following theorem generalizes the construction of a conservative and strongly decisive iterative learner from a strongly locking iterative learner in \cite[Theorem 8]{jain2016role}.
With it we obtain in the Corollary thereafter, that all non-semantic learning restrictions coincide.

\begin{thm}
\label{StronglyLockingImpliesWb}
Let $\CalL$ be a set of languages $\BMS_\ast\Ex$-learned by a strongly locking $\BMS$-learner. Then $$\CalL \in [\Txt\BMS_\ast\Wb\Ex].$$
\end{thm}
\begin{proof}
Let $\CalL \in [\Txt\BMS_\ast\Ex]$ be learned by the strongly locking learner $M$.
By Lemma~\ref{StateDecisive} we may assume that $M$ does not return to withdrawn states.

We proceed in two steps. First we construct a learner $M'$ conservatively $\BMS_\ast\Ex$-learning at least $\CalL$ in a strong sense, i.e.,
\begin{equation}
\label{eq:Def_SConv}
\forall \sigma \in\Seq{\Sigma} \:\forall x\in\Sigma \:(\, M'(\sigma\concat x)\neq M'(\sigma) \;\Rightarrow\; x\notin W_{M'(\sigma)}\,).
\end{equation}
That we require the last datum to violate consistency with the former hypothesis fits the setting of $\BMS$-learners and is also called locally conservative by \cite{Jai-Lan-Zil:c:06}.
Second, with such a learner at hand, we are going to construct a learner $N$ which $\BMS_\ast\Ex$-learns $\CalL$ in a witness-based fashion.
We will do this by keeping track of all data having caused a mind-change so far.
More concretely, we alter the text by excluding mind-change data causing another mind-change and make sure that the witness for the mind-change is contained in all future hypotheses.

\medskip
For defining the strongly conservative learner $M'$, we employ a one-one function $f:\N\times Q\to\Omega$ satisfying
$$
W_{f(e,s)}=\bigcup_{t\in\N} \begin{cases}
W_e^t, &\text{if } \forall x\in W_e^t (\,h_M(s,x)=e\:\wedge\:s_M(s,x)=s\,);\\
\varnothing, &\text{otherwise}
\end{cases}
$$
for every hypothesis $e\in\N\subseteq\Omega$ and state $s\in Q$.
The existence of $f$ is granted by the smn theorem.
Thus, $f$ takes into account only the initial part of $W_e$ not necessary to possibly justify a mind-change or state-change later on.
Now define for all $\sigma\in\Seq{\Sigma}$
\begin{align*}
M'(\sigma)=f(M(\sigma),s^\ast_M(\sigma)).
\end{align*}
As $M$ never returns to withdrawn states and behaves strongly locking while $\BMS_\ast\Ex$-learning $\CalL$, $M'$ also $\Ex$-learns $\CalL$.
For $\sigma\neq\epsilon$ the values of $M(\sigma)$ and $s^\ast_M(\sigma)$ only depend on $s^\ast_M(\sigma^-)$ and $\last(\sigma)$ and hence $M'$ is a $\BMS_\ast$-learner with $s_{M'}=s_M$.
Moreover, by construction it is conservative in the strong sense defined in ($\ref{eq:Def_SConv}$).

\medskip
We now define the witness-based learner $N$. In addition to thinning out the hypotheses of $M'$, as we did with the hypotheses of $M$ when constructing $M'$ from $M$, we patch all data causing mind-changes to it. This data is stored in the states used by $N$. Further, we only alter our old hypothesis in case we can guarantee the existence of a witness justifying the possible mind-change. To do this in a computable way, we need to store also the last hypothesis of $M'$ in the states of $N$.

For every datum $x\in\Sigma$, data-sequence $\sigma\in\Seq{\Sigma}$, hypothesis $e\in\N\subseteq \Omega$ and every finite sequence $\MC$ of natural numbers, interpreted as pairs of hypotheses and data, 
we define a state transition function $s_N$, 
auxiliary hypothesis generating function $M$, 
recursive function $g:\N^2\to\Omega$ and the learner $N$ by
\begin{align*}
h(\langle s, \langle\MC\rangle \rangle,x)&=
\begin{cases}
h_{M'}(s,\#), &\text{if } x\in \pr_2[\MC];\\
h_{M'}(s,x), &\text{otherwise;}
\end{cases}\\
s_N(\langle s,\langle\MC\rangle\rangle,x)&=
\begin{cases}
\langle s_{M'}(s,\#), \langle\MC\rangle\rangle, &\text{if } x\in \pr_2[\MC] \;\wedge\; h_{M'}(s,\#)= \pr_1(\last(\MC));\\
\langle s_{M'}(s,\#), \langle\MC\concat \langle h_{M'}(s,\#),\#\rangle\rangle\rangle, &\text{if } x\in \pr_2[\MC] \;\wedge\; h_{M'}(s,\#)\neq \pr_1(\last(\MC));\\
\langle s_{M'}(s,x), \langle\MC\rangle\rangle,  &\text{else if } h_{M'}(s,x)= \pr_1(\last(\MC));\\
\langle s_{M'}(s,x), \langle\MC\concat \langle h_{M'}(s,x),x\rangle\rangle\rangle, &\text{otherwise;}
\end{cases} \\
W_{g(e,\langle s,\langle \MC \rangle\rangle)} &= \pr_2[\MC] \cup W_e;\\
N(\sigma\concat x)&=
\begin{cases}
?, &\text{if } h^\ast(\sigma\concat x)=\:?;\\
g(h^\ast(\sigma\concat x),s_N^\ast(\sigma\concat x)), &\text{else if } h^\ast(\sigma\concat x)\neq \pr_1(\last(\mathrm{decode}(\pr_2(s^\ast_N(\sigma)))))); \\
N(\sigma), &\text{otherwise.}
\end{cases}
\end{align*}
Thus with the help of $g$ the data stored in the second coordinates of $\MC$ is patched to the language encoded in $e$. Further, $N$ only makes a mind-change if $h^\ast$ does, as $h^\ast(\sigma)=\pr_1(\last(\mathrm{decode}(\pr_2(s^\ast_N(\sigma))))))$. The learner $h^\ast$ behaves like $M'$ on the text, in which every datum repeatedly causing a mind-change is replaced by the pause symbol.

Let $L\in\CalL$ and $T\in\Txt(L)$. It is easy to see that for the text $T'$ recursively defined by
\begin{align*}
T'(t)&=
\begin{cases}
\#, &\text{if } \exists s<t \:(\, T(s)=T(t)\:\wedge\: M'(T'[s]\concat T(s))\neq M'(T'[s]) \,);\\
T(t), &\text{otherwise,}
\end{cases}
\end{align*}
holds $h^\ast(T[t])=M'(T'[t])$ for all $t\in\N$. This follows with a simultaneous induction also showing $\pr_1(s_N^\ast(T[t]))=s_{M'}^\ast(T'[t])$.
Hence $h^\ast$ on $T$ behaves like $M'$ on $T'\in\Txt(L)$.

Because $M'$ $\Ex$-converges on $T'$, it makes only finitely many mind-changes and uses only finitely many states, which implies that $N$ also only uses finitely many states.
Let $e=M'(T'[t_0])$ be the final correct hypothesis of $M'$ on $T'$ with $t_0\in\N$ chosen appropriately. Because $M'$ never returns to withdrawn states, the states of $N$ also stabilize.
Moreover, $N(T[t_0])$ has to be correct since $\pr_2[\MC]\subseteq W_e$.

As already mentioned, $N$ learns every $L\in\CalL$ witness-based because $M'$ is strongly conservative. Every time $N$ performs a mind-change on $T$, so does $M'$ on $T'$. Therefore, there is a responsible datum $x$ which was not in the former hypothesis of $M'$ and also has not occured so far, as no datum in $T'$ causes more than one mind-change. This datum $x$ will be contained in all languages hypothesized by $N$ in the future.
\end{proof}

With the latter theorem it is straightforward to observe that in the $\BMS_\ast\Ex$-setting conservative, strongly decisive and strongly non-U-shaped $\Ex$-learning are equivalent.

\begin{cor}
\label{BMS-ConvSDecSNU}
We have $[\Txt\BMS_\ast\Conv\Ex]=[\Txt\BMS_\ast\SDec\Ex]=[\Txt\BMS_\ast\SNU\Ex]$.
\end{cor}
\begin{proof}
On the one hand a conservative or strongly decisive learning behavior is also a strongly non-U-shaped learning behavior.
On the other hand, a learner behaving strongly non-U-shaped proceeds strongly locking and, by Theorem~\ref{StronglyLockingImpliesWb}, from a strongly locking learner we may construct a learner with at least equal learning power, acting witness-based and hence also conservatively and strongly decisively.
\end{proof}


By \cite[Theorem~2]{jain2016role} and Lemma~\ref{BMS-It_sem-afsoet} (\ref{BMSIt}) we obtain
$$[\Txt\BMS_\ast\Conv\Ex]\not\subseteq[\Txt\BMS_\ast\SMon\Ex].$$
From this we conclude with Theorem~\ref{BMS-SMonSNU} and Corollary~\ref{BMS-ConvSDecSNU} the following incomparability
$$[\Txt\BMS_\ast\Conv\Ex]\perp [\Txt\BMS_\ast\SMon\Ex].$$
 
Similarly, with \cite[Theorem~3]{jain2016role} and again Lemma~\ref{BMS-It_sem-afsoet} (\ref{BMSIt}) we obtain $[\Txt\BMS_\ast\Conv\Ex]\not\subseteq[\Txt\BMS_\ast\Mon\Ex]$. As Theorem~\ref{BMS-SMonSNU} implies $[\Txt\BMS_\ast\Mon\Ex]\not\subseteq[\Txt\BMS_\ast\SNU\Ex]$, with Corollary~\ref{BMS-ConvSDecSNU} follows $$[\Txt\BMS_\ast\Conv\Ex]\perp [\Txt\BMS_\ast\Mon\Ex].$$

\bigskip
Because Theorem~\ref{BMS-SMonSNU} also reproves $[\Txt\BMS_\ast\SNU\Ex]\subsetneq[\Txt\BMS_\ast\Ex]$, first observed in \cite[Th.~3.10]{Cas-Koe:j:11:memLess}, we completed the map for $\BMS_\ast\Ex$-learning from texts. An overview is depicted in Figure~\ref{fig:TxtBMSEx}.

\begin{figure}[h]
\begin{center}
\begin{tikzpicture}[above,sloped,shorten <=2mm,, shorten >=2mm]

\node[shape=rectangle,rounded corners,fill=lightgray] at (-5,-1) {$\Txt\BMS_\ast\Ex$};
\drawBackboneTwo

\node[line width=3pt, rounded corners, fit=(nu) (nothing) (caut) (wmon) (dec)] (Fit1) {};

\node[line width=3pt, rounded corners, fit=(snu) (sdec) (conv)] (Fit2) {};

\path[draw=lightgray,line width=3pt, rounded corners] (Fit1.north west) -- ([yshift=-3mm] Fit1.south west) -- ([xshift=-20mm,yshift=-3mm] Fit1.south east) -- ([xshift=3mm,yshift=20mm] Fit1.south east)-- ([xshift=3mm] Fit1.north east) -- cycle;

\path[draw=lightgray,line width=3pt, rounded corners] (Fit2.south east) -- (Fit2.north east) -- ([xshift=20mm] Fit2.north west) -- ([yshift=-20mm] Fit2.north west)-- (Fit2.south west) -- cycle;

\node[draw=lightgray,line width=3pt, rounded corners, minimum width= 15mm,minimum height = 8mm, anchor=center] at (mon) {};
\node[draw=lightgray,line width=3pt, rounded corners, minimum width= 15mm,minimum height = 8mm, anchor=center] at (smon) {};

\end{tikzpicture}
\end{center}
\caption{Relations between delayable learning restrictions in explanatory finitely bounded memory states learning of languages from informants. The arrows represent implications independent of the model. The outlined areas stand for equivalence classes with respect to learning power, when the underlying model is $\Txt\BMS_\ast\Ex$.}
\label{fig:TxtBMSEx}
\end{figure}
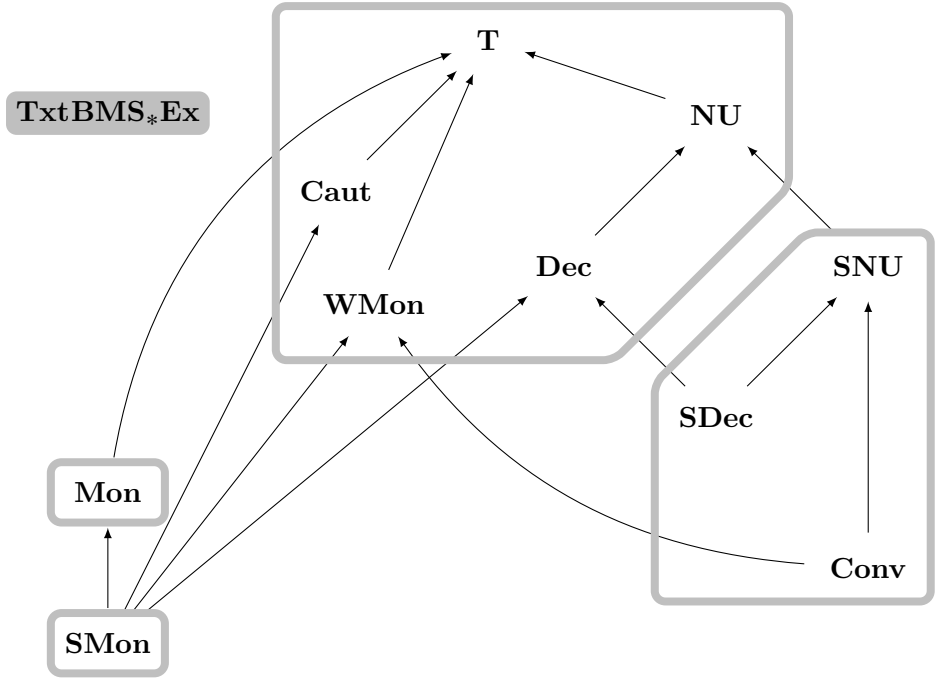

\bigskip
As this map equals the one for $\It$-learning, naturally the question arises, whether a result similar to Lemma~\ref{BMS-It_sem-afsoet} can be observed for the syntactic learning criteria. In the following we show that this is not the case.

\begin{thm}\label{thm:itBMSDifferSNU}
$[\It\Txt\SNU\Ex] \subsetneq [\Txt\BMS_\ast\SNU\Ex]$
\end{thm}
\begin{proof}
By Lemma~\ref{BMS-It_sem-afsoet} we have $[\It\Txt\SNU\Ex] \subseteq [\Txt\BMS_\ast\SNU\Ex]$.

We consider the $\BMS$-learner $M$ initialized with state $\langle \langle \,?,0 \rangle,\langle\varnothing\rangle \rangle$ and $h_M$ and $s_M$ for every $\langle e,\xi \rangle \in\Omega$, $D\subseteq\N$ finite and $x\in\Sigma$ defined by:
\begin{align*}
s_M(\langle \langle e,\xi\rangle,\langle D \rangle \rangle,x)&=
\begin{cases}
\langle \langle e,\xi\rangle,\langle D\rangle \rangle, 												&\text{if } x\in D\cup\{\#\} \;\vee\; \\
& \hspace{2ex} \pr_1(\,\varphi_x(\langle e,\xi\rangle)\!\downarrow\,\,)= e;\\
\langle \varphi_x(\langle e,\xi \rangle), \langle D\cup\{x\}\rangle\rangle, 	&\text{else if } \pr_1(\,\varphi_x(\langle e,\xi \rangle)\!\downarrow\,\,)\neq e;\\
\uparrow, 																																		&\text{otherwise.}
\end{cases} \\
h_M(\langle \langle e,\xi\rangle,\langle D \rangle\rangle,x)&= \begin{cases}
e, &\text{if } x\in D\cup\{\#\} \;\vee\; \\
& \hspace{2ex} \pr_1(\,\varphi_x(\langle e,\xi\rangle)\!\downarrow\,\,)= e;\\
\pr_1(\,\varphi_x(\langle e,\xi\rangle)\,), &\text{else if } \pr_1(\,\varphi_x(\langle e,\xi\rangle)\!\downarrow\,\,)\neq e;\\
\uparrow, &\text{otherwise.}
\end{cases}
\end{align*}
Additionally to the last hypothesis as well as exactly the data that already lead to a mind-change of $M$, some parameter $\xi$ is stored, indicating whether a further mind-change may cause a syntactic $U$-shape.

Let $\CalL=\Txt\BMS_\ast\SNU\Ex(M)$. 
We will show that there is no iterative learner $\It\Txt\SNU\Ex$-learning $\CalL$.
Assume $N$ is an iterative learner with hypothesis generating function $h_N$ and $\CalL\subseteq\It\Txt\Ex(N)$.

We obtain $L\in \CalL\setminus\It\Txt\SNU\Ex(N)$ by applying 1-1 ORT \cite{case1974periodicity} referring to the $\Sigma_1$-predicates $\MC$ and $\NoMC$, 
expressing that $N$ does (not) perform a mind-change on a text built from parameters $a,b\in\totalCp$. More specifically, the predicates state that $N$ does converge and (not) make a mind-change when observing $\sigma \in \Seq{\Sigma}$ after having observed $a[i]\concat b(i)\concat \#^{\ell_i}$, with $i\in\N$.
\begin{align*}
\psi_{i}(\ell) &\Leftrightarrow N(a[i]\concat b(i)\concat \#^{\ell})= N(a[i]\concat b(i)\concat \#^{\ell+1});\\
\NoMC(i,\sigma) &\Leftrightarrow \exists {\ell_i}\in\N\:(\,\psi_{i}(\ell_i) \:\wedge\: \forall \ell<\ell_i \:\neg \psi_i(\ell) \:\wedge\: \\
& \hspace{12.5ex} N(a[i]\concat b(i)\concat \#^{{\ell_i}}\concat \sigma)\!\downarrow \;= N(a[i]\concat b(i)\concat \#^{{\ell_i}})\,); \\
\MC(i,\sigma) &\Leftrightarrow \exists {\ell_i}\in\N\:(\,\psi_{i}(\ell_i) \:\wedge\: \forall \ell<\ell_i \:\neg \psi_i(\ell) \:\wedge\: \\
& \hspace{12.5ex} N(a[i]\concat b(i)\concat \#^{{\ell_i}}\concat \sigma)\!\downarrow \;\neq N(a[i]\concat b(i)\concat \#^{{\ell_i}}) \,).
\end{align*}

By 1-1 ORT, applied to the recursive operator implicit in the following case distinction, there are recursive total functions $a,b,e_1,e_2$ with pairwise disjoint ranges and $e_0\in\N$, such that for all $i,\xi\in\N$, $e\in\Omega$
\begin{align*}
\varphi_{a(i)}(\langle e,\xi \rangle) & =
\begin{cases}
	\langle e_0,\xi \rangle, 				&\text{if } e\in\{?,e_0\};\\
	\langle e_1(k),1 \rangle,				&\text{else if } \xi = 0, i\text{ even and }\exists k\leq i \,(\,e=e_1(k)\,);\\
	\langle e_1(k),2 \rangle,				&\text{else if } \xi = 0, i\text{ odd and }\exists k\leq i \,(\,e=e_1(k)\,);\\
	\langle e_2(k),0 \rangle,				&\text{else if } \xi = 1, i\text{ odd and }\exists k\leq i \,(\,e=e_1(k)\,);\\
	\langle e_2(k),0 \rangle,				&\text{else if } \xi = 2, i\text{ even and }\exists k\leq i \,(\,e=e_1(k)\,);\\
	\langle e,\xi \rangle, 					&\text{otherwise;}
\end{cases} \\
\varphi_{b(i)}(\langle e, \xi\rangle)&=\begin{cases}
\langle e_1(i),\xi\rangle, &\text{if } e\in\{?,e_0\};\\
\langle e,\xi\rangle, &\text{otherwise;}\\
\end{cases}\\
W_{e_0}&=\begin{cases}
\ran(a[t_0]), &\text{if $t_0$ is minimal with } \forall t\geq t_0\, N(a[t])=N(a[t_0]); \\
\ran(a), &\text{no such $t_0$ exists};
\end{cases}\\
W_{e_1(i)}&=
\ran(a[i])\cup\{b(i)\}\cup\begin{cases}
\{a(j)\} &\text{for first } j \geq i \text{ found} \\
&\text{with }\MC(i,a(j));\\
\varnothing, &\text{no such $j$ exists;}
\end{cases} \\
W_{e_2(i)}&=
\ran(a)\cup\{b(i)\}.
\end{align*}

As the learner constantly puts out $e_0$ on every text for $W_{e_0}$, we have $W_{e_0}\in\CalL$.
Thus, also $N$ learns the finite language $W_{e_0}$ and $t_0$ exists. Note that by the iterativeness of $N$ we obtain $N(a[t_0])=N(a[t_0]\concat a(i)) \text{ for all } i \geq t_0$ and with this
$N(a[t_0]\concat b(t_0)\concat \#^{\ell_{t_0}})=N(a[t_0]\concat a(i)\concat b(t_0)\concat \#^{\ell_{t_0}}) \text{ for all } i \geq t_0.$

$W_{e_1(t_0)}$ and $W_{e_2(t_0)}$ also lie in $\CalL$. 
To see that $M$ explanatory learns both of them, note that, after having observed $b(t_0)$, $M$ only changes its mind from $e_1(t_0)$ to $e_2(t_0)$ after having seen $a(i)$ and $a(j)$ with $i,j\geq t_0$ and $i \in 2\N$ as well as $j \in 2\N+1$. This clearly happens for every text for the infinite language $W_{e_2(t_0)}$. As $|W_{e_1(t_0)}\setminus \left( \cnt(a[t_0])\cup\{b(t_0)\} \right)| \leq 1$, this mind change never occurs for any text for $W_{e_1(t_0)}$.

The syntactic non-U-shapedness of $M$'s learning processes can be easily seen as for all $k, l \in \N$ the languages $W_{e_0}$, $W_{e_1(k)}$ and $W_{e_2(l)}$ are pairwise distinct, the learner never returns to an abandoned hypothesis and $M$ only leaves hypothesis $\langle e_1(k),0\rangle$ for $\langle e_1(k),\xi\rangle$, $\xi\neq 0$, if $W_{e_1(k)}$ is not correct. 

Next, we show the existence of $j \geq t_0$ with $\MC(t_0,a(j))$.
Assume towards a contradiction that $j$ does not exist.
Then $W_{e_1(t_0)}=\cnt(a[t_0])\cup\{b(t_0)\}$.
As $M$ learns this language from the text $a[t_0]\concat b(t_0)\concat\#^\infty$, so does $N$.
The convergence of $N$ implies the existence of $\ell_{t_0}$.
Thus, for every $j\in\N$ we either have $N(a[t_0]\concat b(t_0)\concat\#^{\ell_{t_0}}\concat a(j))=N(a[t_0]\concat b(t_0)\concat\#^{\ell_{t_0}})$ or the computation of $N(a[t_0]\concat b(t_0)\concat\#^{\ell_{t_0}}\concat a(j))$ does not terminate.
Because $N$ is iterative and learns $W_{e_2(t_0)}$, it may not be undefined and therefore always the latter is the case.
But then $N$ will not learn $W_{e_1(t_0)}$ and $W_{e_2(t_0)}$ as they are different but $N$ does not make a mind-change on the text $a[t_0]\concat b(t_0)\concat \#^{\ell_{t_0}} \concat a$ after having observed the initial segment $a[t_0]\concat b(t_0)\concat \#^{\ell_{t_0}}$, due to its iterativeness.
Hence, $j$ exists and $W_{e_1(t_0)}=\ran(a[t_0])\cup\{b(t_0),a(j)\}$.

Finally, by the choice of $j$, the learner $N$ does perform a syntactic U-shape on the text $a[t_0]\concat a(j)\concat b(t_0)\concat \#^{\ell_{t_0}}\concat a(j)\concat \#^\infty$ for $W_{e_1(t_0)}$.
More precisely, $t_0$ and $\ell_{t_0}$ were chosen such that $N(a[t_0]\concat a(j)\concat b(t_0)\concat \#^{\ell_{t_0}})$ has to be correct and the characterizing property of $j$ assures $$N(a[t_0]\concat a(j)\concat b(t_0)\concat \#^{\ell_{t_0}})\neq N(a[t_0]\concat a(j)\concat b(t_0)\concat \#^{\ell_{t_0}}\concat a(j)).$$
Thus, no iterative learner can explanatory syntactically non-U-shapedly learn the language $\CalL$.
\end{proof}

\smallskip
By Corollary~\ref{BMS-ConvSDecSNU} we also obtain $[\It\Txt\SDec\Ex] \subsetneq [\Txt\BMS_\ast\SDec\Ex]$ and $[\It\Txt\Conv\Ex] \subsetneq [\Txt\BMS_\ast\Conv\Ex]$.

\section{Related Open Problems}

We have given a complete map for learning with bounded memory states, where, on the way to success, the learner must use only finitely many states.
Future work can address the complete maps for learning with an a priori bounded number of memory states, which needs very different combinatorial arguments.
Results in this regard can be found in \cite{Car-Cas-Jai-Ste:j:07} and \cite{Cas-Koe:j:11:memLess}.
We expect to see trade-offs, for example allowing for more states may make it possible to add various learning restrictions (just as non-deterministic finite automata can be made deterministic at the cost of an exponential state explosion).

\medskip
Also memory-restricted learning from positive and negative data (so-called informant) has only partially been investigated for iterative learners and to our knowledge not at all for other models of memory-restricted learning.
Very interesting also in regard of 1-1 hypothesis spaces that prevent coding tricks is the $\mathbf{Bem}$-hierarchy, see \cite{Ful-Jai-Osh:j:94:osw}, \cite{Lan-Zeu:j:96} and \cite{Cas-Jai-Lan-Zeu:j:99:feedback}.

\medskip
In the spirit of grammatical inference, we encourage to investigate the learnability of carefully chosen indexable families arising from applied machine learning or cognitive science research.

\subsection*{Acknowledgements}
This work was supported by DFG Grant Number KO 4635/1-1. We are grateful to the people supporting us.

\bibliographystyle{alpha}

\bibliography{CLTBib}

\newcommand{\etalchar}[1]{$^{#1}$}
\begin{thebibliography}{MPU{\etalchar{+}}92}

\bibitem[AKS18]{As-Koe-Sei2018_informants}
M.~Aschenbach, T.~K{\"o}tzing, and K.~Seidel.
\newblock Learning from informants: Relations between learning success
  criteria.
\newblock {\em arXiv preprint arXiv:1801.10502}, 2018.

\bibitem[Ang80]{angluin1980inductive}
D.~Angluin.
\newblock Inductive inference of formal languages from positive data.
\newblock {\em Information and control}, 45(2):117--135, 1980.

\bibitem[BB75]{Blu-Blu:j:75}
L.~Blum and M.~Blum.
\newblock Toward a mathematical theory of inductive inference.
\newblock {\em Information and Control}, 28:125--155, 1975.

\bibitem[BCM{\etalchar{+}}08]{Bal-Cas-Mer-Ste-Wie:j:08}
G.~Baliga, J.~Case, W.~Merkle, F.~Stephan, and R.~Wiehagen.
\newblock When unlearning helps.
\newblock {\em Information and Computation}, 206:694--709, 2008.

\bibitem[Cas74]{case1974periodicity}
J.~Case.
\newblock Periodicity in generations of automata.
\newblock {\em Mathematical Systems Theory}, 8(1):15--32, 1974.

\bibitem[Cas94]{Cas:j:94:self}
J.~Case.
\newblock Infinitary self-reference in learning theory.
\newblock {\em Journal of Experimental and Theoretical Artificial
  Intelligence}, 6:3--16, 1994.

\bibitem[CC13]{Car-Cas:j:13:topics}
L.~Carlucci and J.~Case.
\newblock On the necessity of {U}-shaped learning.
\newblock {\em Topics in Cognitive Science}, 5:56--88, 2013.
\newblock Invited for Special Issue on Formal Learning Theory; see {\tt
  dx.doi.org/10.1111/tops.12002} for html form.

\bibitem[CCJS07]{Car-Cas-Jai-Ste:j:07}
L.~Carlucci, J~Case, S.~Jain, and F.~Stephan.
\newblock Results on memory-limited {U}-shaped learning.
\newblock {\em Information and Computation}, 205:1551--1573, 2007.

\bibitem[CJLZ99]{Cas-Jai-Lan-Zeu:j:99:feedback}
J.~Case, S.~Jain, S.~Lange, and T.~Zeugmann.
\newblock Incremental concept learning for bounded data mining.
\newblock {\em Information and Computation}, 152:74--110, 1999.

\bibitem[CK10]{Cas-Koe:c:10:colt}
J.~Case and T.~K{\"o}tzing.
\newblock Strongly non-{U}-shaped learning results by general techniques.
\newblock In Adam~Tauman Kalai and Mehryar Mohri, editors, {\em {COLT} 2010},
  pages 181--193, 2010.

\bibitem[CK13]{Cas-Koe:j:11:memLess}
J.~Case and T.~K{\"o}tzing.
\newblock Memory-limited non-u-shaped learning with solved open problems.
\newblock {\em Theoretical Computer Science}, 473:100--123, 2013.

\bibitem[CK16]{case2016strongly}
J.~Case and T.~K{\"o}tzing.
\newblock Strongly non-u-shaped language learning results by general
  techniques.
\newblock {\em Information and Computation}, 251:1--15, 2016.

\bibitem[CM08]{Cas-Moe:j:08:NUIt}
J.~Case and S.~Moelius.
\newblock U-shaped, iterative, and iterative-with-counter learning.
\newblock {\em Machine Learning}, 72:63--88, 2008.

\bibitem[CM11]{Cas-Moe:j:11:optLan}
J.~Case and S.~Moelius.
\newblock Optimal language learning from positive data.
\newblock {\em Information and Computation}, 209:1293--1311, 2011.

\bibitem[FJO94]{Ful-Jai-Osh:j:94:osw}
M.~Fulk, S.~Jain, and D.~Osherson.
\newblock Open problems in {S}ystems {T}hat {L}earn.
\newblock {\em Journal of Computer and System Sciences}, 49(3):589--604,
  December 1994.

\bibitem[Gol67]{Gol:j:67}
E.~Gold.
\newblock Language identification in the limit.
\newblock {\em Information and Control}, 10:447--474, 1967.

\bibitem[Jan91]{j-mniifp-91}
K.~P. Jantke.
\newblock Monotonic and nonmonotonic inductive inference of functions and
  patterns.
\newblock In {\em Nonmonotonic and Inductive Logic, 1st International Workshop,
  Proc.}, pages 161--177, 1991.

\bibitem[JKMS16]{jain2016role}
S.~Jain, T.~K{\"o}tzing, J.~Ma, and F.~Stephan.
\newblock On the role of update constraints and text-types in iterative
  learning.
\newblock {\em Information and Computation}, 247:152--168, 2016.

\bibitem[JLZ06]{Jai-Lan-Zil:c:06}
S.~Jain, S.~Lange, and S.~Zilles.
\newblock Towards a better understanding of incremental learning.
\newblock In {\em ALT}, volume 4264 of {\em Lecture Notes in Computer Science},
  pages 169--183, 2006.

\bibitem[JMZ13]{Jai-Moe-Zil:j:13}
S.~Jain, S.~Moelius, and S.~Zilles.
\newblock Learning without coding.
\newblock {\em Theoretical Computer Science}, 473:124--148, 2013.

\bibitem[JORS99]{Jai-Osh-Roy-Sha:b:99:stl2}
S.~Jain, D.~Osherson, J.~Royer, and A.~Sharma.
\newblock {\em Systems that Learn: {A}n Introduction to Learning Theory}.
\newblock MIT Press, Cambridge, Massachusetts, second edition, 1999.

\bibitem[K{\"o}t09]{Koe:th:09}
T.~K{\"o}tzing.
\newblock {\em Abstraction and Complexity in Computational Learning in the
  Limit}.
\newblock PhD thesis, University of Delaware, 2009.

\bibitem[KP16]{kotzing2016map}
T.~K{\"o}tzing and R.~Palenta.
\newblock A map of update constraints in inductive inference.
\newblock {\em Theoretical Computer Science}, 650:4--24, 2016.

\bibitem[KS16]{kotzing2016towards}
T.~K{\"o}tzing and M.~Schirneck.
\newblock Towards an atlas of computational learning theory.
\newblock In {\em 33rd Symposium on Theoretical Aspects of Computer Science},
  2016.

\bibitem[KSS17]{KSS17}
T.~Kötzing, M.~Schirneck, and K.~Seidel.
\newblock Normal forms in semantic language identification.
\newblock In {\em Proc.~of Algorithmic Learning Theory}, pages 493--516. PMLR,
  2017.

\bibitem[LZ96]{Lan-Zeu:j:96}
S.~Lange and T.~Zeugmann.
\newblock Incremental learning from positive data.
\newblock {\em Journal of Computer and System Sciences}, 53:88--103, 1996.

\bibitem[MPU{\etalchar{+}}92]{Mar:b:92}
G.~Marcus, S.~Pinker, M.~Ullman, M.~Hollander, T.J. Rosen, and F.~Xu.
\newblock {\em Overregularization in Language Acquisition}.
\newblock Monographs of the Society for Research in Child Development, vol.~57,
  no.~4. University of Chicago Press, 1992.
\newblock Includes commentary by H. Clahsen.

\bibitem[Odi99]{Odi:b:99}
P.~Odifreddi.
\newblock {\em Classical Recursion Theory}, volume~II.
\newblock Elsivier, Amsterdam, 1999.

\bibitem[OSW82]{Osh-Sto-Wei:j:82:strategies}
D.~Osherson, M.~Stob, and S.~Weinstein.
\newblock Learning strategies.
\newblock {\em Information and Control}, 53:32--51, 1982.

\bibitem[OSW86]{STL1}
D.~Osherson, M.~Stob, and S.~Weinstein.
\newblock {\em Systems that Learn: {A}n Introduction to Learning Theory for
  Cognitive and Computer Scientists}.
\newblock MIT Press, Cambridge, Mass., 1986.

\bibitem[RC94]{Roy-Cas:b:94}
J.~Royer and J.~Case.
\newblock {\em Subrecursive Programming Systems: Complexity and Succinctness}.
\newblock Research monograph in {\em Progress in Theoretical Computer Science}.
  {Birkh\"auser}~Boston, 1994.

\bibitem[SS82]{Str-Sta:b:82}
S.~Strauss and R.~Stavy, editors.
\newblock {\em U-Shaped Behavioral Growth}.
\newblock Developmental Psychology Series. Academic Press, NY, 1982.

\bibitem[Wie91]{Wie:c:91}
R.~Wiehagen.
\newblock A thesis in inductive inference.
\newblock In {\em Nonmonotonic and Inductive Logic, 1st International Workshop,
  Proc.}, pages 184--207, 1991.

\end{thebibliography}

\end{document}